\documentclass{amsart}
\usepackage[cp1251]{inputenc}
\usepackage{amsmath, amsfonts}
\usepackage[pdffitwindow = true, pdfstartview = FitH]{hyperref}

\theoremstyle{plain}
\newtheorem{theorem}{Theorem}
\newtheorem{lemma}[theorem]{Lemma}
\newtheorem{corollary}[theorem]{Corollary}

\theoremstyle{definition}
\newtheorem{definition}[theorem]{Definition}
\newtheorem{example}[theorem]{Example}
\theoremstyle{remark}

\newtheorem{question}{Question}

\newcommand{\R}{\mathbb{R}}  
\newcommand{\Z}{\mathbb{Z}}  
\newcommand{\N}{\mathbb{N}}  
\DeclareMathOperator{\conv}{conv} 
\DeclareMathOperator{\ext}{ext} 
\DeclareMathOperator{\xc}{xc} 
\DeclareMathOperator{\rc}{rc} 
\newcommand{\eps}{\varepsilon}
\DeclareMathOperator{\M}{\eps} 
\DeclareMathOperator{\CP}{CP} 
\DeclareMathOperator{\ES}{S} 
\DeclareMathOperator{\Faces}{\mathcal F} 
\DeclareMathOperator{\Lat}{\mathcal L} 

\newcommand{\BQP}{\textup{BQP}} 
\newcommand{\CBQP}{\textup{CBQP}} 

\title{Complexity of LP in Terms of the~Face Lattice}

\author{Aleksandr Maksimenko}
\address{Laboratory of Discrete and Computational Geometry, P.G. Demidov Yaroslavl State University, ul. Sovetskaya 14, Yaroslavl 150000, Russia} 
\email{maximenko.a.n@gmail.com}

\usepackage{fancyhdr}
\thanks{Supported by the~project No. 477 of P.\,G.~Demidov Yaroslavl State University within State Assignment for~Research.}

\begin{document}

\begin{abstract}
Let $X$ be a~finite set in $\Z^d$.
We consider the~problem of optimizing linear function $f(x) = c^T x$ on $X$, 
where $c\in\Z^d$ is an~input vector.
We call it a~\emph{problem} $X$.
A problem $X$ is related with linear program $\max\limits_{x \in P} f(x)$, 
where polytope $P$ is a~convex hull of $X$.
The~key parameters for evaluating the~complexity of a~problem $X$ are
the~dimension $d$, the~cardinality $|X|$, 
and the~encoding size $\ES(X) = \log_2 \left(\max\limits_{x\in X} \|x\|_{\infty}\right)$.
We show that if the~(time and space) complexity of some algorithm $A$ for solving a~problem $X$
is defined \emph{only} in terms of combinatorial structure of $P$ and the~size $\ES(X)$, 
then for every $d$ and $n$ there exists polynomially (in $d$, $\log n$, and $\ES$) solvable problem $Y$ with $\dim Y = d$, $|Y| = n$, 
such that the~algorithm $A$ requires exponential time or space for solving $Y$.
\end{abstract}

\maketitle

\section{Introduction}

In many cases a~combinatorial optimization problem can be stated in the~following form.

\medskip

\textsc{Given} a~finite set of feasible solutions $X \subset \Z^d$,
and a~linear function $f(x) = c^T x$, $c \in \Z^d$, $x \in X$.

\textsc{Find} the~maximum (minimum) value of $f(x)$.

\medskip

We will call it a~\emph{problem} $X$, assuming an~arbitrary choice of the~input vector $c \in \Z^d$.
For example, in the~travelling salesman problem the~set $X$, $X\subset \{0,1\}^{E}$ is
the~set of characteristic vectors of hamiltonian circuits in a~graph $G = (V, E)$.

A problem $X$ is related with linear program (LP)
\begin{equation*}
\max_{x \in P} c^T x, \qquad \text{where } P = \conv X .
\end{equation*}
This is the~main reason of interest to such geometric statement 
of a~combinatorial optimization problem.

It is clear that the~complexity of a~problem $X$ may depend on the~encoding size%
\footnote{This is the~natural requirement for the~modern digital devices.} 
\begin{equation}
\label{eq:size}
\ES(X,c) = \log_2 \left(\max_{x\in X\cup\{c\}} \|x\|_{\infty}\right).
\end{equation}
But $\ES(X,c)$ does not reflect the~structural complexity of $X$.

So, it is natural to consider some combinatorial characteristics
of $P = \conv X$ as characteristcs of complexity of a~problem $X$.
The simplest examples are the~dimension of $P$, the~number of its vertices, 
and the~number of its facets.
Nontrivial examples are the~diameter of the~graph (1-ske\-le\-ton) of $P$, 
the~clique number of the~graph, and the~rectangle covering number 
of the~ver\-tex-fa\-cet (non)incidence matrix.

The diameter of the~graph of $P$ was considered as the~lower bound
of complexity of a~problem\footnote{$X$ is the~set of extreme points of $P$} $X = \ext P$ in the~class of sim\-plex-ty\-pe algorithms.
The weakness of this bound is illustrated by the~following well known example.
For any (arbitrary complicated) polytope $P$ 
one can consider a~pyramid $Q$ with $P$ as a~base.
It is obvious that the~problem $Y = \ext Q$ is not simpler 
than the~problem $X = \ext P$, but the~graph diameter of a~pyramid 
is not greater than 2.

In~1980's, V.\,A.~Bondarenko introduced the~concept of so-cal\-led direct type
algorithms~\cite{Bondarenko:1995, Bondarenko:2008}.
The main idea is that the~clique number of the~graph of $\conv X$ is the~lower bound
on the~complexity of the~appropriate problem $X$ in this class of algorithms.
We~discuss this theory and its limitations in the~section~\ref{sec:direct-type}.
In~particular, we show that there is no an~algorithm whose complexity 
for solving a~problem $X$ is expressed only 
in the~clique number of the~graph of $\conv X$ and in the~encoding size $\ES(X,c)$.

Polytope $Q$ is called an~\emph{extension} (or extended formulation) of a~polytope $P$ 
if there is a~linear projection $\pi$ with $\pi(Q) = P$.
In this context, the~number of facets of a~polytope is frequently called a~\emph{size} of a~polytope.
It is well known that the~size of an~extension $Q$ may be significantly less than the~size of its projection $P$.
On the~other hand the~problem $\max\limits_{x\in P} c^T x$ is easily reduced to the~problem $\max\limits_{y\in Q} b^T y$.
Thus in some cases it is usefull to express a~polytope $P$ via its extension.
The minimum size of an~extension of a~polytope $P$ is called \emph{extension complexity} of $P$.

In the~end of~1980's, M.~Yannakakis in his seminal paper~\cite{Yannakakis:1991} on extended formulations 
noticed that the~extension complexity is bounded from below 
by the~rectangle covering number of the~ver\-tex-fa\-cet (non)incidence matrix.
Several breakthrough results was obtained in this direction over the~past three years 
(see \cite{Fiorini:2012}, \cite{Fiorini:2013}, \cite{Rothvoss:2014}).
All of them suggest that there may exist an~algorithm whose complexity for solving
a problem $X$ is equal to big O (or some polynomial) of the~rectangle covering number.
In~the~section~\ref{sec:rectcover} we enumerate this facts and show 
that there is no such algorithm.

Our main result is presented in the~section~\ref{sec:main}.
Let the~function $f$ takes each problem $X$ to $\N$.
We assume that $f$ is defined only in terms of the~face lattice of $\conv X$ and the~encoding size $\ES(X,c)$,
and $f$ is monotone in $\ES(X,c)$ and the~face lattice (by embedding).
(I.e. the~function $f$ is a~monotone combinatorial characteristic of complexity of $X$.)
Then there are an~exponentially solvable problem $Y$ and a~polynomially solvable problem $Z$
such that $f(Y) \le f(Z)$.

%
%

\section{Direct Type Algorithms}
\label{sec:direct-type}

\subsection{Introduction to the~theory}
The information in this subsection is not crucial for the~rest of the~paper,
but it seems that there is no description of the~theory of direct type algorithms in English.
So we have to say ``a couple of words'' about this interesting theory~\cite{Bondarenko:1995, Bondarenko:2008}.

First of all we should say that direct type algorithms are \emph{linear search algorithms} (LSAs).
When dealing with LSAs, one takes into account only time necessary for branchings of the~form
 ``if $f(c) >  0$ then goto $\alpha$, elese goto $\beta$''~\cite{Meyer:1984}.
Here $c \in \R^d$ is the~input vector of a~problem $X$ 
 and $f(c) = a^T c + b$ is an~affine function, 
 where $a \in \R^d$, $b \in \R$.
It is convenient to imagine the~structure of an~LSA as a~\emph{linear decision tree} (LDT)
 with tests ``$f(c) >  0$'' in internal nodes.
Every terminal node (leaf) of such tree has some label $x \in X$.
(One label $x \in X$ can be assigned to more than one leaf.)
 
The \emph{complexity $C_{\text{LSA}}(X)$ of a~problem $X$}
 is the~minimum depth of an~LDT for $X$.
Clearly, $C_{\text{LSA}}(X) \ge \log |X|$
 (provided that for every $x \in X$ there exists an~input $c$ 
 s.t. $x$ is the~optimal solution).
In \cite{Moshkov:1982} (see also \cite{Moshkov:2005}) the~upper bound $O(d^3 \log |X|)$
 have been found for $C_{\text{LSA}}(X)$.
But such LSA can occupy an~exponential space.
So, it would be good to add some natural restrictions to LSA model.

In 1980th, V.A. Bondarenko introduced the~concept of so-called \emph{direct type algorithms} \cite{Bondarenko:1995}.
This concept is based on the~notion of partition of $\R^d$ into cones.
For every $x \in X$ we consider its cone
$$
 K(x) = \{c \in \R^d \mid c^T x  \ge c^T y , \ \forall y \in X\}.
$$
The set of all such cones for given $X$ is called 
 a~\emph{partition of $\R^d$ into cones w.r.t. $X$}.
It is clear, that $K(x)$ is not empty i\textbf{ff} $x$ lies on the~boundary of the~convex hull $\conv X$.
In particular, if $X\subseteq \{0, 1\}^d$ then every $x \in X$ is a~vertex of $\conv X$ 
 and, hence, every $K(x)$ has interior points.
Traditionally, the~convex hull $\conv X$ is called \emph{the~polytope of a~problem $X$}.
Typically, the~set $X$ coincides 
 with the~set of vertices of the~polytope $\conv X$.
Here and below we assume the~latter condition is fulfilled.
I.e., for every $x \in X$ there exists $c \in K(x)$ such that 
 $c^T x  > c^T y $ for any $y \in X$, $y \ne x$.

Two cones $K(x)$ and $K(y)$ are called \emph{adjacent} if 
$$
  \exists c \in K(x)\cap K(y) \ : \ c^T x = c^T y > c^T z
  \quad \forall z \in X, \ z \ne x, \ z \ne y.
$$
It is obvious that $K(x)$ and $K(y)$ are adjacent i\textbf{ff} the~vertices $x$ and $y$
 of the~polytope $\conv X$ are adjacent (i.e., these vertices form a~1-face of the~polytope).
The subset $Y \subseteq X$ is called \emph{a clique in $X$} if every pair $\{x, y\} \subseteq Y$
 is adjacent.

Let $T$ is a~linear decision tree for a~problem $X$ and
 $f$ is some internal node of $T$.
Let $L(f)$ is the~set of leaves of $T$ that are descendants of $f$.
We denote by $X_f$, $X_f \subseteq X$, the~set of labels of leaves $L(f)$.
The set $L(f)$ is divided into two parts $L^+(f)$ and $L^-(f)$
 by two arcs outgoing from $f$.
Let $X^+_f$ and $X^-_f$ denote all the~labels of leaves $L^+(f)$ and $L^-(f)$, respectively.

 
\begin{definition}[\cite{Bondarenko:1995}] 
\label{defDTA}
A linear decision tree for solving a~problem $X$ is called \emph{direct type tree}
 if for any internal node $f$ and for any clique $Y \subseteq X$ the~following 
 inequality holds:
$$
 |X_f \cap Y| - 1 \le \max\{|X^+_f \cap Y|, |X^-_f \cap Y|\}.
$$ 
\end{definition}

The \emph{complexity $C_{\text{DTT}}(X)$ of $X$} is the~minimum depth 
 of a~direct type tree for a~problem $X$.

Let $\omega(X)$ is the~size of the~maximum clique in $X$.
I.e. $\omega(X)$ is the~clique number of the~1-skeleton of $\conv X$.
It is not difficult to see that
\begin{equation*}
  C_{\text{DTT}}(X) \ge \omega(X) - 1.
\end{equation*}

It is known \cite{Bondarenko:1995, Bondarenko:2008}, that 
 sorting algorithms, 
 greedy algorithms for the~minimum spanning tree, 
 Dijkstra's algorithm for the~shortest path in a~graph,
 Held-Karp algorithm and branch-and-bound algorithms for the~travelling salesman problem,
 and some other combinatorial algorithms are direct type algorithms.
On the~other hand, clique numbers $\omega(X)$ are superpolynomial
 for such NP-hard problems as the~travelling salesman, the~knapsack, 
 the~3-satisfiability, the~3-assignment, the~maximum cut, 
 the~set covering, the~set packing and many others \cite{Bondarenko:1995, Maksimenko:2013, Maksimenko:2014}.
Whereas $\omega(X)$ are polynomial for polynomially solvable problems: 
 the~sorting, the~minimum spanning tree, the~short path in a~graph, the~min-cut problem
 \cite{Bondarenko:1995, Nikolaev:2013}.

Nonetheless, there are examples of polynomially solvable problems $X$ 
with exponential $\omega(X)$ \cite{Bondarenko:1995}.
A generalization of one such example is considered in the~following subsection.

It is also natural to ask the~following question.

\begin{question}[V. Kaibel]
Is there some (NP-)hard problems with small $\omega(X)$?
\end{question}

%
%

\subsection{Examples}

The first example is related with the~famous cyclic polytopes.
(More detailed information on cyclic polytopes is presented in~\cite{Grunbaum:2003}.)
Let us consider a~monotone function
\begin{equation}
\label{eq:sequence}
g:\ \N \to \N, \qquad g(i) < g(i+1) \quad \forall i\in\N,
\end{equation}
and the~set 
\begin{equation}
\label{eq:cyclic}
C(d,N,g) = \left\{(t_i, t_i^2, \dots, t_i^d)\in\N^d \mid t_i = g(i),\ i\in[N]\right\},
\end{equation}
where $d,N\in\N$.
The convex hull of $C(d,N,g)$ is a~cyclic polytope.
It is well known that $d$-di\-men\-sio\-nal cyclic polytope is a~simplicial
$\lfloor d/2 \rfloor$-neigh\-borly polytope and it has the~maximum number of faces
among all convex polytopes with the~same number $N$ of vertices~\cite{Grunbaum:2003}.
In particular, the~clique number $\omega(C(d,N,g)) = N$ for $d \ge 4$.

\begin{theorem}
\label{th:cyclic}
A problem $C(d,N,g)$ is solvable polynomially in $d$ and $\log N$ whenever
the~function $g$ is polynomially countable 
and the~encoding size of an~input vector $c$
is polynomial in $d$ and $\log N$.
In particular, if $g(i) = i$, $i\in\N$ and $\|c\|_{\infty} = O(N)$, 
then the~solving of a~problem $C(d,N,g)$ requires $O(d^4 \log^2 N)$ time and $O(d \log N)$ space.
\end{theorem}

\begin{proof}
We have to maximize function $c^T x = c_1 x_1 + c_2 x_2 + \dots + c_d x_d$ for $x\in C(d,N,g)$. 
That is we have to maximize the~polynomial
\begin{equation*}
f(t)	= c_1 t + c_2 t^2 + \dots + c_d t^d, \quad \text{where } t = g(i), \ i\in[N].
\end{equation*}
The algorithm of finding the~maximum will consist of $d-1$ steps.

In the~first step we divide the~set $[N]$ into (two or one) segments 
where the~derivative $f^{(d-1)}(t) = (d-1)!\, c_{d-1} + d!\, c_d t$ has constant sign.

In the~second step we consider the~derivative 
\[
f^{(d-2)}(t) = (d-2)!\, c_{d-2} + \frac{(d-1)!}{1!} c_{d-1} t + \frac{d!}{2!} c_d t^2.
\]
It is monotone in every segment found in the~previous step.
Hence it is not difficult to divide every such segment into segments
with constant sign of $f^{(d-2)}(t)$.
It can be done by dichotomic procedure and requires no more than $2\log_2 N$ evaluatings of $f^{(d-2)}(t)$.

In the~following steps, by analogy we eventually partition the~set $[N]$ into 
at most $d$ segments with constant sign of $f'(t)$.
Thus, it remains to find the~value of $f(t)$ at the~ends of these segments and choose the~maximum.

To conclude the~proof it remains to note that the~calculation of $f'(t)$ requires $O(d)$ arithmetic operations
 with $(d \log g(N))$-bit numbers. In particular, it takes $O(d^2 \log N)$ time for the~case $g(i) = i$.
\end{proof}

As a~consequence, for every $k\in\N$ the~problem $C(2k, 2^k, g)$ with $g(i)=i$
is solvable polynomially in $k$, but the~clique number $\omega\bigl( C(2k, 2^k, g) \bigr) = 2^k$
is exponential. Moreover, $\conv C(2k, 2^k, g)$ is $k$-neighborly.
Hence, any its $k$ vertices form a~(simplicial) face.

It turns out that there are also examples of the~opposite nature.
More precisely, there are problems with $\omega(X) = 2$ and arbitrary complexity.

\begin{theorem}
\label{th:simplex}
For any simplex $\Delta_d\subset\R^d$ there is an~extension $Q\subset\R^{d+1}$ with $\omega(Q) = 2$.
\end{theorem}

\begin{proof}
We use the~fact that any two $d$-dimensional simplices are affinely equivalent to each other.
So, we consider only the~``simplest'' one:
\[
\Delta_{d-1} = \{x\in \R_+^d \mid x \in H\},
\]
which is the~intersection of a~nonnegative orthant $\R_+^d$ and a~hyperplane
\[
H = \{x=(x_1, \dots, x_d) \in \R^d \mid x_1 + \dots + x_d = 1\}.
\]
We will construct an~extension $Q \subset \R^d$ such that the~projection of $Q$ 
into the~hyperplane $H$ coincides with $\Delta_{d-1}$.
Therefore, $Q$ is contained in the~cylinder
\[
 Y = \biggl\{x=(x_1, \dots, x_d) \in \R^d \Bigm| (n-1) x_i + 1 \ge \sum_{j\ne i} x_j, \ i\in[d]\biggr\}.
\]

Besides, let $Q$ be symmetric with respect to $H$.
So, we construct only half of $Q$ that lies in $H^+ = \{x \in \R^d \mid x_1 + \dots + x_d \ge 1\}$.
This half of $Q$ we denote by $Q^+$.
Let $Q^+$ be the~intersection of the~cylinder $Y$, 
halfspace $H^+$, and the~unit cube $C_d = \{x\in\R_+^d \mid x \le \mathbf{1}\}$.
It is not difficult to see that the~combinatorial structure of $Q^+$ is
the~combinatorial structure of the~``cube without one vertex'' $H^+ \cap C_d$.
All vertices of $Q^+$ can be divided into $d$ groups according to the~number of coordinates equal to 1.
The first group consists of one vertex $(1, 1, \dots, 1)$.
Every vertex in the~second group has one coordinate $\frac{d-2}{d-1}$ and $d-1$ ones.
Every vertex in the~third group has two coordinates $\frac{d-3}{d-2}$ and $d-2$ ones.
The last group consists of $d$ basis vectors.
% ...

We remark that every triangle in the~graph of the~polytope $Q^+$ is contained in $H$
or has one edge in $H$.
Thus, when we glue $Q^+$ and $Q^-$ all triangles will disappear.
Hence $\omega(Q) = 2$.
\end{proof}

It is obvious that the~linear optimization on the~simplex $\Delta_d$ 
(in the~general case) requires at least $d^2$ operations.
The same is true for the~optimization on its extension $Q$.
Thus, in this example the~dimension of the~problem 
characterizes the~complexity much more accurately, than the~clique number.
Moreover, the~theorem says that there is no an~algorithm whose complexity
for solving a~problem $X$ is expressed only in the~clique number $\omega(X)$ and encoding size of $X$.

%
%

\section{Rectangle Covering Numbers}
\label{sec:rectcover}

\subsection{Background}

Let $V = \{v_1, \dots, v_n\}$ be the~set of vertices of a~polytope $P$
and $F = \{F_1, \dots, F_k\}$ be the~set of its facets.
The vertex-facet nonincidence matrix $A=(a_{ij})\in \{0,1\}^{n\times k}$ of $P$
 is defined as follows:
\[
a_{ij} = \begin{cases}
0, & \text{if } v_i \in F_j,\\
1, & \text{otherwise.}
\end{cases}
\]
Let $I\subseteq [n]$, $J\subseteq [k]$.
The set $I\times J$ is called \emph{1-rectangle} in $A$ if $a_{ij} = 1$ for all $i\in I$ and $j\in J$.
A \emph{rectangle covering} of $A$ is the~set of 1-rectangles whose union is equal to union of 1-entries in $A$.
The \emph{rectangle covering number} of $A$ is the~smallest cardinality of a~rectangle covering of $A$.
By following \cite{Fiorini:2013} we denote by $\rc(P)$ the~rectangle covering number of the~vertex-facet nonincidence matrix of $P$.

It is known, that the~rectangle covering number is the~lower bound on an~extension complexity of $P$ \cite{Yannakakis:1991}
(see \cite{Fiorini:2013} for the~current knowledge on this topic).
In particular, if some problem $X$ (more precisely, its convex hull $\conv X$) has a~compact extended formulation,
then $\rc(\conv X)$ is polynomial.
At the~present time, there are known a~lot of polynomially solvable problems with compact extended formulations.
Among them are sorting problems, spanning trees, matchings, cuts, 
approximation case of the~knapsack problem and many others \cite{Conforti:2010}.
Special mention should be the~perfect matching polytope $P_{\text{M}}(n)$.
It has a~polynomial rectangle covering number $\rc(P_{\text{M}}(n)) = O(n^4)$ \cite{Fiorini:2013},
but exponential extension complexity $\xc(P_{\text{M}}(n)) = 2^{\Omega(n)}$ \cite{Rothvoss:2014}.
On the~other hand, the~boolean quadratic polytope (correlation polytope) $\BQP_n$ 
has exponential rectangle covering $\rc(\BQP_n) = 2^{\Omega(n)}$ \cite{Fiorini:2012} 
(see also \cite{Kaibel:2013} for the~best current bound).
Consequence of this result are superpolynomial lower bounds on rectangle covering numbers 
for many other NP-hard problems involving travelling salesman problem, knapsack problem,
satisfiability problems, 3-assignment problem, set covering and set packing problems,
and many others \cite{Fiorini:2012, Maksimenko:2013, Maksimenko:2014}.
These facts let one to conjecture that the~rectangle covering number 
is a~complexity of some algorithm (or class of algorithms) for solving a~problem $X$.
It turns out that this is not true.
We show that there are NP-hard problems with polynomial $\rc$.

Our construction is based on the~fact that $\rc(P)$ of a~simplicial polytope $P$
is equal to $O(d^2 \log n)$~\cite{Fiorini:2013}, where $d = \dim P$ is the~dimension 
and $n = |\ext P|$ is the~number of vertices of $P$.
The main idea is to make a~slight perturbation of vertices of 0/1-polytope
associated with NP-hard problem.

%
%

\subsection{Cyclic Perturbation}

For every $x\in \{0, 1\}^d$ we define its number 
\[
n(x) = \sum_{i=1}^d 2^{i-1} x_i, \quad 0 \le n(x) \le 2^d - 1.
\]
Let us consider a~map $\M: \{0, 1\}^d \to \N^d$ for transforming $x\in \{0, 1\}^d$ to $\eps \in \N^d$:
\begin{align*}
\eps_1 &= n(x), \\
\eps_2 &= (n(x))^2, \\
       & \dots\\
\eps_d &= (n(x))^d.
\end{align*}
Let 
\[
K = 2^{d^3}
\]
be ``very big'' constant.
Note that for any $x\in \{0, 1\}^d$, the $\|\M(x)\|$ is ``very small'' with respect to $K$:
\[
\frac{\|\M(x)\|}{K} \le \frac{\|\M(x)\|_1}{K} \le \frac{(2^d-1) + (2^d-1)^2 + \dots + (2^d-1)^d}{K} \le \frac{2^{d^2}}{2^{d^3}} 
 = 2^{-d^2(d-1)}.
\]

Let $X\subseteq \{0, 1\}^d$.
The set 
\[
Y = \CP(X) = \bigl\{y\in\Z^d \mid y = K x + \M(x), \ x\in X\bigr\} 
\]
is called \emph{cyclic perturbation} of $X$.
It is clear that after such perturbation the~encoding size of $X$ increases in $\log_2 K = d^3$ times.
Furthermore, $Y = \ext \conv Y$ since the~value $\|\M(x)\|$ of perturbation is ``very small''.

\begin{lemma}
The convex hull of the~cyclic perturbation of $X\subseteq \{0,1\}^d$ is a~simplicial polytope.
\end{lemma}

\begin{proof}
It is sufficient to prove that any $d+1$ points\footnote{We consider only interesting cases when $|\CP(X)| \ge d+1$. 
The other cases can be proved by analogy.} 
in the~cyclic perturbation 
\[
Y = \CP(X)
\]
are affinely independent.

For every $(d+1)$-point set $\{y^1, y^2, \dots, y^{d+1}\} \subseteq Y$ we must check that
\begin{equation}
\label{eq:matrix}
\det\begin{vmatrix}
1 & y^1_1 & y^1_2 & \!\dots & y^1_d \\[2pt]
1 & y^2_1 & y^2_2 & \!\dots & y^2_d \\[2pt]
\hdotsfor[1]{5} \\[2pt]
1 & y^{d+1}_1 & y^{d+1}_2 & \!\dots & y^{d+1}_d
\end{vmatrix} \ne 0.
\end{equation}
Since $y^i = K x^i + \M(x^i)$ for some $x^i \in X$, $i\in[d+1]$,
we may decompose the~matrix \eqref{eq:matrix} into the~sum of two matrices
\[
A = \begin{pmatrix}
0 & K x^1_1 & K x^1_2 & \!\dots & K x^1_d \\[2pt]
0 & K x^2_1 & K x^2_2 & \!\dots & K x^2_d \\[2pt]
\hdotsfor[1]{5} \\[2pt]
0 & K x^{d+1}_1 & K x^{d+1}_2 & \!\dots & K x^{d+1}_d
\end{pmatrix} \in \{0, K\}^{(d+1)\times(d+1)}
\]
and
\[
B = \begin{pmatrix}
1 & n(x^1) & (n(x^1))^2 & \!\dots & (n(x^1))^d \\[2pt]
1 & n(x^2) & (n(x^2))^2 & \!\dots & (n(x^2))^d \\[2pt]
\hdotsfor[1]{5} \\[2pt]
1 & n(x^{d+1}) & (n(x^{d+1}))^2 & \!\dots & (n(x^{d+1}))^d \\[2pt]
\end{pmatrix}.
\]
For every set $S\subseteq[d+1]$ we define $(d+1)\times (d+1)$-matrix $D^S$:
\[
D^S_{ij} = 
\begin{cases}
A_{ij},& \text{if } i\in S,\\
B_{ij},& \text{if } i\not\in S.
\end{cases}
\]
In particular, $D^{[d+1]} = A$, $D^{\emptyset} = B$.
As is known, the~determinant of the~sum of two $n\times n$-matrices can be written 
as the~sum of determinants of $2^n$ matrices:
\begin{equation}
\label{eq:twomatrixsum}
\det(A+B) = \sum_{S\subseteq[d+1]} \det D^S.
\end{equation}
For every nonempty set $S$ the~matrix $D^S$ has at least one row from $A$.
Therefore, $\det D^S$ is a~multiple of $K$.
On the~other hand, $\det D^{\emptyset} = \det B$ is the~Vandermonde determinant:
\[
\det B = \prod_{1\le j < k \le d+1} \bigl(n(x^k) - n(x^j)\bigr) > 0.
\]
Moreover,
\[
\det B = \prod_{1\le j < k \le d+1} \bigl(n(x^k) - n(x^j)\bigr) \le (2^d-1)^{d(d+1)/2} < K.
\]
Hence, the~sum \eqref{eq:twomatrixsum} is not equal to 0.
\end{proof}

Since $\rc(P) = O(d^2 \log|\ext P|)$ for any simplicial polytope $P\in \R^d$~\cite{Fiorini:2013}, 
we immediately get

\begin{corollary}
\label{cor:simplicial}
The rectangle covering number $\rc(\conv CP(X)) = O(d^2 \log |X|) = O(d^3)$ for any $X \subseteq\{0,1\}^d$.
\end{corollary}

\begin{example}
Let us consider the~cyclic perturbation of vertices of the~boolean quadratic polytope%
\footnote{Sometimes we use the~same notation for the~polytope and its vertices.}
\begin{equation*}
 \BQP_n = \left\{x=(x_{ij})\in\{0,1\}^{n(n+1)/2} \mid 
	        x_{ij} = x_{ii} x_{jj}, \; 1\le i < j\le n\right\}.
\end{equation*}
Obviously, the~encoding size (in any reasonable sense) of $\BQP_n$ is polynomial in $n$.
Hence, the~same is true for its cyclic perturbation $\CBQP_n = \CP(\BQP_n)$.
Moreover, due to corollary~\ref{cor:simplicial}, 
the~rectangle covering number $\rc(\conv \CBQP_n) = O(n^3 (n+1)^2)$ is polynomial.

Now we show that the~optimization problem $\CBQP_n$ is NP-hard.
Let us consider the~NP-hard problem of finding a~clique number 
in an~undirected graph $G=(V, E)$ with $n$ vertices $V = [n]$.
The input vector $c = c(G) \in \Z^{n(n+1)/2}$ we define as follows:
\[
c_{ij} = \begin{cases}
1, & \text{if } i=j,\\
0, & \text{if } \{i,j\} \in E,\\
-2, & \text{if } \{i,j\} \not\in E.
\end{cases}
\]
It is easy to see that 
$\max\limits_{x\in \BQP_n} c^T x$
is equal to the~clique number of $G$.
But for any $x \in \BQP_n$ and $y = CP(x)$ 
\begin{equation*}
\left|c^T x - c^T y / K\right| = |c^T \M(x) / K|
  \le \frac{2}{2^{d^2(d-1)}} \le \frac{1}{2^{17}},
\end{equation*}
where $d = n(n+1)/2$, $n \ge 2$.
Hence the~problem $\CBQP_n$ is not easy than the~clique number problem.
\end{example}

%
%

\section{The Main Result}
\label{sec:main}

Let us denote by $\Faces(P)$ the~set of all faces of a~polytope $P$
without improper faces $\emptyset$ and $P$.
The \emph{face lattice} $\Lat(P)$ of a~polytope $P$ 
is the~set $\Faces(P)$ (partially) ordered by inclusion.
The face lattice $\Lat(P)$ of a~polytope $P$ is \emph{embedded} into the~face lattice $\Lat(Q)$ of a~polytope $Q$
if there is a~map $h: \Faces(P) \to \Faces(Q)$ such that $F_1 \subseteq F_2$ i\textbf{ff} $h(F_1) \subseteq h(F_2)$
for any two faces $F_1, F_2 \in \Faces(P)$.
In such a case we will use the notation $\Lat(P) \preceq \Lat(Q)$.

In what follows we need the~following example.
Let $P\subset \R^d$ be a~simplicial polytope with $n$ vertices
and let $C(2d, n) = \conv C(2d, n, g)$ be a~cyclic polytope with $n$ vertices (see the~equation \eqref{eq:cyclic}).
Note that $C(2d, n)$ is $d$-neighborly (i.e. every $d$ vertices form a~(simplicial) face of $C(2d, n)$).
Hence $\Lat(P) \preceq \Lat(C(2d, n))$.
Moreover, for such embedding we get a~several additional properties.
The number of vertices and the~number of facets of $P$ do not exceed these numbers of $C(2d,n)$,
the~graph of $P$ is a~subgraph of the~graph of $C(2d,n)$, and so on.

Let $f$ be the~map of face lattices to $\N$.
We say that $f$ is \emph{monotone} if $f(\Lat(P)) \le f(\Lat(C(2d,n)))$ for any simplicial $d$-polytope $P$ with $n$ vertices.
It is natural to assume that combinatorial characteristics of complexity of a~polytope
must satisfy this property.
Here are just a~few examples of such functions:
the~dimension of a polytope,
the~number of its $k$-faces, the~clique number of the~graph, 
the~maximum number of vertices of a~$k$-neigh\-borly face, the~rectangle covering number.
The~monotonicity is not obvious only for the~rectangle covering, but it is proved in~\cite[Corollary 2.13]{Fiorini:2013}.
The diameter of the~graph of a~polytope is an~example of a~nonmonotone function.

In addition, when evaluating complexity of the~problem $X\subset \Z^d$, 
we should take into account its size $\ES(X, c)$ (see equation \eqref{eq:size}).
If a~(natural) function $f(X, c) = f(\Lat(\conv X), \ES(X, c))$ is monotone
in every argument, we call it a~\emph{monotone combinatorial characteristic of complexity} of a~problem $X$.
Below we restrict ourselves to the~cases where
the~size of an~input vector $c$ does not make a~significant contribution to $\ES(X, c)$
and we will write $f(X)$ instead of $f(X, c)$.

\begin{theorem}
\label{th:main}
There exist an~intractable problem $X$ and a~polynomially solvable problem $Y$ such that $f(X) \le f(Y)$
for any monotone combinatorial characteristic of complexity $f$.
\end{theorem}

\begin{proof}
It is well known that almost all Boolean functions require exponential-sized circuits (see, for example, \cite[Chapter 6]{Arora:2009}).
Let $b: \{0,1\}^d \to \{0,1\}$ be such (intractable) function and
\[
Z = \bigl\{z\in\{0,1\}^d \mid b(z)=1\bigr\}.
\]
For each $x\in\{0,1\}^d$ we define the~input vector $c(x)\in\{-1, 1\}^d$ as follows:
\[
c(x)_i = 2 x_i - 1, \quad i\in[d].
\]
It is easy to see that $\max\limits_{z\in Z} (c(x))^T z$ is equal to the~number of ones in $x$
for $b(x) = 1$, otherwise it is less than the~number of ones.
Thus, the~problem $Z$ is intractable even for input $c\in\{-1, 1\}^d$.
The same is true for the~cyclic perturbation $X = \CP(Z)$, since
\[
|c^T z - c^T \CP(z) / K| = |c^T \M(z) / K| \le \frac{1}{2^{d^2(d-1)}} \le \frac1{16} 
\] 
for $c\in\{-1, 1\}^d$, $z\in Z$, and $d \ge 2$.
Besides, the~encoding size $\ES(X) = O(d^3)$ is polynomial and the~polytope $\conv X$ is simplicial.

To finish the~proof we assume
\[
Y = C(2d, |X|, g),
\]
where $g: \N \to \N$ is a~polynomially countable monotone function with $g(|X|) = 2^{\ES(X)}$.
Thus, the~problem $Y$ is polinomially solvable by theorem \ref{th:cyclic}.
\end{proof}

Let us assume that the complexity of some algorithm $A$ for solving a problem $X$
is a monotone combinatorial characteristic of complexity of $X$.
The theorem asserts that such an algorithm requires exponential (time or space) complexity
for solving a problem $C(2d, |X|, g)$, but this problem is polynomially solvable by theorem~\ref{th:cyclic}.
For example, the (polynomial of) rectangle covering number can not be a characteristic of complexity of such algorithm,
since $\rc(C(2d, |X|, g))$ is polynomial.
But it may be considered as a lower bound.
On the other hand, the clique number $\omega(C(2d, |X|, g)) = |X|$ is exponential in $\log |X|$.
However, this characteristic violates more strict condition of monotonicity: 
the implication $\Lat(X) \preceq \Lat(Y)$ $\Rightarrow$ $f(X) \le f(Y)$ should be true not only for $Y = C(2d, |X|, g)$.
The appropriate example is provided by theorem~\ref{th:simplex}.

\section{Concluding Remarks}  

All the~mentioned above combinatorial characteristics of complexity 
(with the~exception of the~diameter of the~graph of a~polytope) are monotone. 
We may ask the~natural question: is the~monotonicity a~necessary condition?
More precisely, are there an~intractable problem $X$ and a polynomially solvable problem $Y$ 
with the~same combinatorial structure and polynomially comparable encoding sizes $\ES(X)$ and $\ES(Y)$?
This question is reduced to the~following one.
Is it true that one of the~projections of a~cyclic polytope $C(2d, |X|, g)$ is combinatorially
isomorphic to a~simplicial polytope $X$ in theorem~\ref{th:main}?

\section*{Acknowledgments}

The author is grateful to Volker Kaibel, G\"unter~M.~Ziegler, Matthias Walter, and
Francisco Santos for helpful discussions.

%
%


\begin{thebibliography}{99}

\bibitem{Arora:2009}
  S. Arora and B. Barak. 
	\emph{Computational complexity: a~modern approach.} 
  Cambridge University Press, 2009.

\bibitem{Bondarenko:1995} 
  V.A. Bondarenko.
	\emph{Graphs of polytopes and complexity in combinatorial optimization.}
  Yaroslavl': YarGU, 1995. 
	(in Russian)

\bibitem{Bondarenko:2008} 
  V.A. Bondarenko, A.N. Maksimenko.
  \emph{Geometric constructions and complexity in combinatorial optimization.}
  Moscow: LKI, 2008. 
  (in Russian)

\bibitem{Nikolaev:2013} 
  Bondarenko, V.A., Nikolaev, A.V. 
  Combinatorial and Geometric Properties of the~Max-Cut and Min-Cut Problems.
  \emph{Doklady Mathematics}, \textbf{88}(2), 516--517, 2013.

\bibitem{Conforti:2010}
  M. Conforti, G. Cornu\'ejols, and G. Zambelli. 
  Extended formulations in combinatorial optimization. 
  \emph{A Quarterly Journal of Operations Research}, 8(1): 1--48, 2010.


\bibitem{Fiorini:2012}
  S. Fiorini, S. Massar, S. Pokutta, H.R. Tiwary, R. de Wolf.
  Linear vs. semidefinite extended formulations: exponential separation 
  and strong lower bounds. \emph{In STOC}, 95--106, 2012.

\bibitem{Fiorini:2013}
  S. Fiorini, V. Kaibel, K. Pashkovich, D.O. Theis.
  Combinatorial Bounds on Nonnegative Rank and Extended Formulations.
  \emph{Discrete Math.}, 313(1): 67--83, 2013.

\bibitem{Grunbaum:2003}
  B. Gr\"unbaum.
	\emph{Convex Polytopes.}
	Second Edition prepared by V. Kaibel, V. Klee, and G.M. Ziegler, 
	volume~221 of Graduate Texts in Mathematics, 2003.

\bibitem{Kaibel:2011}
  V. Kaibel.
  Extended Formulations in Combinatorial Optimization.
  \emph{Optima 85}, 2011.

\bibitem{Kaibel:2013}
  V. Kaibel, S. Weltge.
  A Short Proof that the~Extension Complexity of the~Correlation Polytope Grows Exponentially.
  \emph{arXiv:1307.3543}, 2013.

\bibitem{Maksimenko:2013} 
  A.N. Maksimenko.
  The common face of some 0/1-polytopes with NP-complete nonadjacency relation.
  \emph{Fundamentalnaya i prikladnaya matematika}, \textbf{18}(2), 105--118, 2013. 
	(In Russian)

\bibitem{Maksimenko:2014}
  A.N. Maksimenko.
  A special place of Boolean quadratic polytopes among other combinatorial polytopes. 
	\emph{arXiv:1408.0948}, 2014.

\bibitem{Meyer:1984} 
  F. Meyer auf der Heide.
  A polynomial linear search algorithm for the~$n$-dimensional knapsack problem.
  \emph{J. ACM}, \textbf{31}(3), 668--676, 1984.

\bibitem{Moshkov:1982} 
  M.Ju. Moshkov.
  On conditional tests.
  \emph{DAN SSSR}, \textbf{265}(3), 550--552, 1982. 
	(in Russian)

\bibitem{Moshkov:2005} 
  M.J. Moshkov.
  Time complexity of decision trees.
  \emph{In Transactions on Rough Sets III}, 
  Springer-Verlag, Berlin, Heidelberg, 244--459, 2005.

\bibitem{Rothvoss:2014}
  T. Rothvoss.
  The matching polytope has exponential extension complexity.
  \emph{In STOC}, 263--272, 2014.

\bibitem{Yannakakis:1991} 
  M. Yannakakis.
  Expressing combinatorial optimization problems by linear programs. 
  \emph{J. Comput. Syst. Sci.}, 43(3): 441--466, 1991.

\end{thebibliography}
\end{document}